\documentclass[10pt,conference]{IEEEtran}
\usepackage{amsfonts}
\usepackage{amssymb}
\usepackage{stfloats}
\usepackage{cite}
\usepackage{graphicx}
\usepackage{psfrag}
\usepackage{amsmath}
\usepackage{array}
\usepackage{epstopdf}
\usepackage{authblk}
\usepackage{graphicx} 
\usepackage{amsthm} 
\usepackage{lipsum}
\usepackage{verbatim} 
\usepackage{color}
\usepackage{caption}
\usepackage{subcaption}
\usepackage{amsmath,amsfonts,amssymb,mathrsfs}
\usepackage{bm}
\usepackage{ dsfont }
\usepackage{ upgreek }
\usepackage{ stmaryrd }
\usepackage[linesnumbered,lined,boxed,commentsnumbered]{algorithm2e} 

\newcommand{\SNR}{\Upgamma}

\newcommand{\setK}{ \mathcal{K}}
\newcommand{\state }{      \bm{\mathsf {S} } }
\newcommand{\proj }{      \mathsf{Proj} }
\newcommand{\loc }{      \bm{\ell} }
\newcommand{\om }{      \nu_{2}^{\texttt{OM}} }
\newcommand{\traj }{   \bm{\tau}^{\mathbf{b}, \bm{\lambda}}_{\texttt{MC}  } }

\usepackage{graphicx} 

\newtheorem{Definition}{Definition}

\newtheorem{Lemma}{Lemma}
\newtheorem{Theorem}{Theorem}

\newtheorem{Scheme}{Scheme}

\newtheorem{Remark}{Remark}

\usepackage{titlesec}
\titlespacing\section{0pt}{10pt plus 4pt minus 2pt}{0pt plus 2pt minus 2pt}
\titlespacing\subsection{0pt}{9pt plus 4pt minus 2pt}{0pt plus 2pt minus 2pt}
\titlespacing\subsubsection{0pt}{9pt plus 4pt minus 2pt}{0pt plus 2pt minus 2pt}
\begin{document}
\title{ 
Sensing-Assisted Adaptive Channel Contention for Mobile Delay-Sensitive Communications
}

\author{
Bojie Lv, Qianren Li and Rui Wang \\
Southern University of Science and Technology

}
\maketitle

\begin{abstract}
This paper proposes an adaptive channel contention mechanism to optimize the queuing performance of a distributed millimeter wave (mmWave) uplink system with the capability of environment and mobility sensing. The mobile agents determine their back-off timer parameters according to their local knowledge of the uplink queue lengths, channel quality, and future channel statistics, where the channel prediction relies on the environment and mobility sensing. The optimization of queuing performance with this adaptive channel contention mechanism is formulated as a decentralized multi-agent Markov decision process (MDP). Although the channel contention actions are determined locally at the mobile agents, the optimization of local channel contention policies of all mobile agents is conducted in a centralized manner according to the system statistics before the scheduling. In the solution, the local policies are approximated by analytical models, and the optimization of their parameters becomes a stochastic optimization problem along an adaptive Markov chain. An unbiased gradient estimation is proposed so that the local policies can be optimized efficiently via the stochastic gradient descent method. It is demonstrated by simulation that the proposed gradient estimation is significantly more efficient in optimization than the existing methods, e.g., simultaneous perturbation stochastic approximation (SPSA).

\end{abstract}

\section{Introduction}

Distributed channel access mechanisms, such as carrier sense multiple access (CSMA), are widely utilized in wireless systems, such as wireless fidelity (WiFi) and Zigbee, due to their low signaling overhead. However, they may suffer from low scheduling efficiency, particularly for mobile delay-sensitive communications. This is because the transmission scheduling cannot be promptly adapted to the global queuing and channel states. This paper will exploit the wireless sensing and adaptive channel contention design to relieve the above issue.

Extensive studies have been devoted to the centralized delay-sensitive scheduling design using the Markov decision process (MDP) \cite{RuiWang2013, RuiWang-2011TSP}, or Lyapunov optimization \cite{Neely2006TIT}. However, these centralized approaches might not directly apply to the WiFi-like distributed channel access mechanism.
Recently, several research efforts have focused on developing adaptive channel contention mechanisms. For example, the authors in \cite{Modiano_Infocom2023} modified the back-off timer parameters to replicate the behavior of the near-optimal centralized policy, such that the weighted sum of age-of-information (AoI) in a single-hop wireless network could be suppressed. Similarly, the works \cite{Walrand_2010_ToN,Jiang2012-TIT,Srikant_2012ToN} replicated the behavior of near-optimal centralized policies by adjusting the CSMA back-off timer parameters according to the present network state, such as queue lengths. It was shown that the performance of centralized scheduling schemes could be approached with a certain probability with these adaptive CSMA mechanisms. However, the above works did not address the distributed channel contention design from the optimization point of view. Moreover, the impact of mobility on delay-sensitive communications was not investigated in these works. 

The mobility of mobile agents might lead to unexpected channel quality fluctuation from the perspective of queuing performance. Fortunately, wireless sensing technology can predict wireless channel fluctuations due to mobility.
For instance, authors in \cite{li2022indoor,Xinyu_Zhang_NSDI_2017} demonstrated that mmWave communication transceivers can sense the indoor layout, enabling signal strength prediction in arbitrary locations. Additionally, researchers in \cite{yu2023mmalert} employed passive sensing technology to predicate the LoS link blockage during mmWave communication.
However, there is still no existing work on the sensing-assisted channel contention mechanism design.

This paper will shed some light on the above issues by proposing an optimization framework for adaptive channel contention in a sensing-enabled mmWave mobile system. Remarkably, several mobile agents contend the uplink channel according to their local knowledge of the queue lengths and locations. Moreover, the environment and mobility sensing could enable mobile agents to estimate the current channel quality and predict future channel statistics. The joint optimization of local channel contention policies of mobile agents is formulated as a decentralized multi-agent MDP, with the objective of the average system queuing performance. A novel stochastic gradient descent (SGD) method is then proposed to optimize the local policies according to the system statistics before scheduling. Finally, the simulation demonstrates the superior efficiency and performance of the proposed SGD method compared to the existing benchmarks.

This paper uses the following notations: Bold lowercase $\mathbf{a}$ represents column vectors.
Bold uppercase $\mathbf{A}$ represents matrices.
Non-bold letters such as $a$ and $A$ denote scalar values.
The letters $\mathcal{A}$ and $\mathscr{A}$ represent sets.
The magnitude of a scalar is denoted by $\left|a\right|$.
$(a)^{+}$ denotes $\max(0,a)$.
$\left[\mathbf{A}\right]_{i,j}$ represents the $(i,j)$-th element of matrix $\mathbf{A}$.
$\mathbf{A}^\top$ denotes the transpose of matrix $\mathbf{A}$.
$\mathbf{A}^{\mathsf{H}}$ denotes the conjugate transpose of matrix $\mathbf{A}$.  $(a_{j})_{j}$ with $j\in \mathcal{J}$ denotes the column vector whose entries' indexes take values from sets $\mathcal{J}$ in ascending order.
$\mathds{1}\left[.\right]$ represents an indicator function, which takes the value $1$ when the event is true and $0$ otherwise.

\section{System Model}\label{Sec:System_Model}

\begin{figure}[t]
	\centering
	\includegraphics[height=125pt,width=225pt]{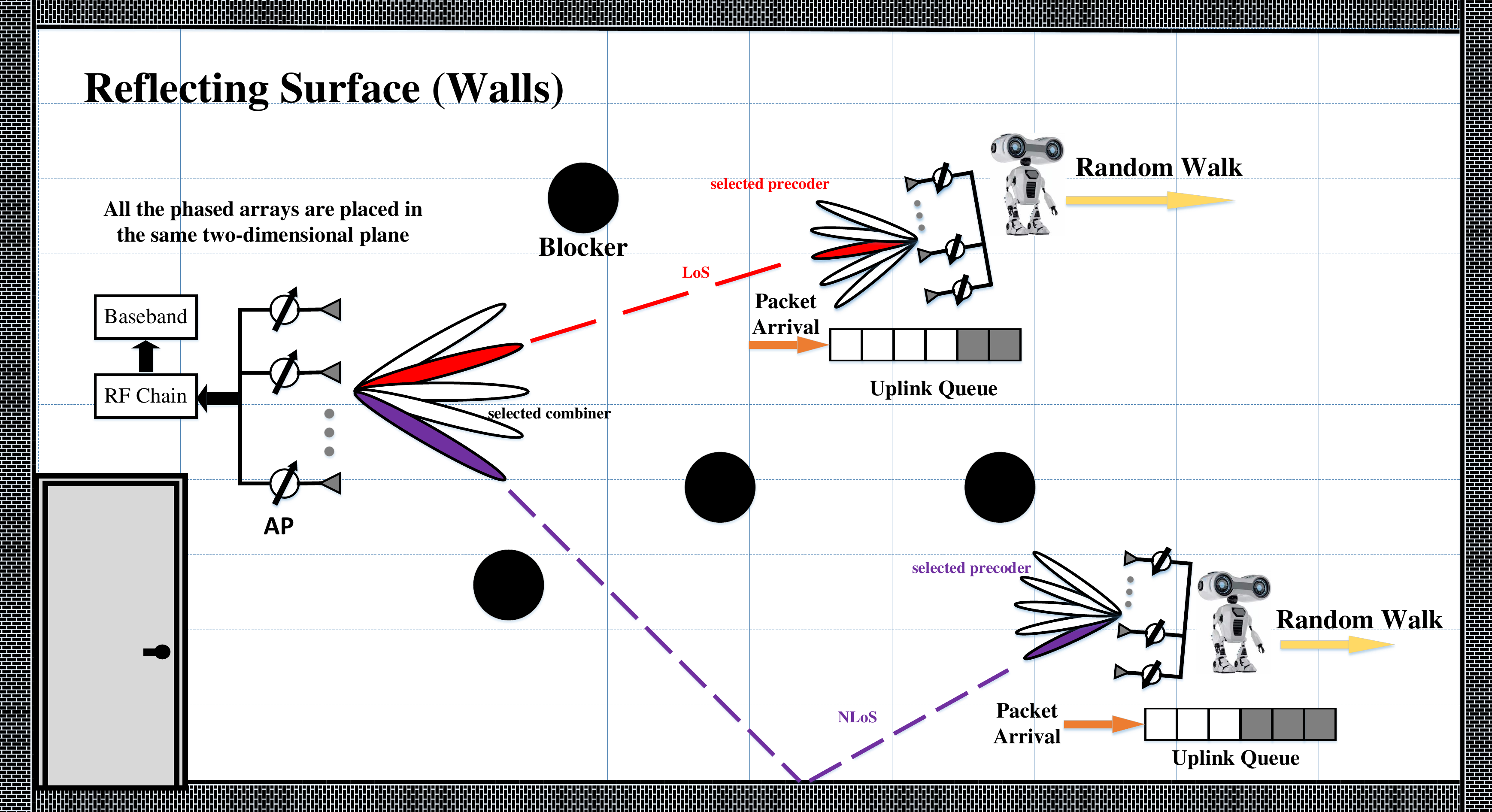}
	\caption{Illustration of uplink communication scenarios: mobile agents are navigating in a room with furniture, pillars as signal blockers; the walls are the main scattering clusters.}
	\label{fig:system_model}
 	\vspace{-0.6cm}
\end{figure}

The uplink transmission scheduling of a WiFi-like mmWave communication system with distributed channel contention is investigated. The system consists of an access point (AP) and $K$ mobile agents, denoted as $\setK \triangleq \{1,2,\ldots,K \}$. These mobile agents move randomly within a two-dimensional plane $\mathscr{S}$, having random packet arrivals for uplink transmission. The mmWave communications are sensitive to the link blockage. Significant channel fluctuation could degrade the delay-sensitive communications. Fortunately, the high-resolution environment sensing \cite{Xinyu_Zhang_NSDI_2017,li2022indoor} and localization \cite{Tsai_2020} provided by mmWave signals can relieve the above issue. Conventional distributed channel contention mechanisms, e.g., CSMA, fail to exploit the sensing information. In this paper, the mobile agents are designed to determine their uplink channel contention strategies distributively and respectively according to local channel status, future channel prediction, and data payload, such that the mobile agents with urgent data would have higher priority in channel contention and the signaling overhead of the centralized scheduling design can be eliminated. In the following, the mobility model, channel contention model, uplink channel and transmission model, as well as the uplink queuing model are elaborated respectively.

\subsection{Mobility Model}

To model the mobility, the possible locations of mobile agents are quantized into a set of discrete location points, as illustrated in Fig. \ref{fig:system_model}. Let $\mathcal{L}\triangleq \left\{ \loc^{(1)}, \loc^{(2)}, \dots, \loc^{\left|\mathcal{L}\right|}\right\}$ with $\loc^{(i)}\in \mathscr{S}$ ($\forall i \in \left\{ 1,2,\ldots, \left|\mathcal{L}\right| \right \} $) represent the coordinates of these location points. The uplink transmission time is organized as a sequence of slots. The locations of mobile agents are quasi-static in one slot and change randomly in the next slot. Denote $\loc _{t,k}^{\texttt A} \in \mathcal{L}$ as the location of the $k$-th mobile agent in the $t$-th slot. It is assumed that the mobility patterns of agents adhere to a time-invariant Markov process, characterized by the following transition probabilities
$
\mathds{P}\left[{\bm \ell}_{t+1,k}^{\texttt A}=\loc^{(j)}\Big|{\bm \ell}_{t,k}^{\texttt A}=\loc^{(i)} \right]= \left[\mathbf{P}_k^{ \texttt{Loc}}\right ]_{ij}$, 
where $\mathbf{P}_k^{ \texttt{Loc}}\in \mathbb{R}^{|\mathcal{L}|\times |\mathcal{L}|}$ is refereed to as the location transition matrix of the $k$-th mobile agent. The mobile agents can detect their locations via the existing sensing techniques \cite{yu2023mmalert}. Moreover, before scheduling the uplink transmission, the location transition matrices for all mobile agents can be estimated at the AP using methods outlined in \cite{Nicholson-Mobicom-08}.

The above motion of mobile agents results in the fluctuation of uplink transmission. The analog multiple-input multiple-output (MIMO) transceivers, which consist of a single radio frequency (RF) chain and a half-wavelength uniform linear phased array (ULA), are deployed at both the AP and the mobile agents. The ULA at the AP and the mobile agents are with $N_{\texttt{R}}$ and $N_{\texttt{T}}$ antenna elements, respectively. Consequently, the AP and the mobile agents can adapt to align their analog receive and transmit beams. Due to the obstacles in the channel, LoS paths may not be available when the mobile agents move to some locations of the plane as illustrated in Fig. \ref{fig:system_model}, and NLoS paths via scattering clusters can be aligned in this case. The absence of LoS paths results in a lower signal-to-noise ratio (SNR). 

\subsection{Distributed Channel Contention Model}

For the elaboration convenience, the idealized channel contention model widely used in the existing literature  \cite{Modiano_Infocom2023,Walrand_2010_ToN} is adopted in this paper to determine the transmission mobile agent for each slot. Particularly, all the mobile agents sense the channel availability for a random duration since the very beginning of the transmission slot (i.e., random back-off). The random durations of the mobile agents follow independent exponential distributions with heterogeneous expectations. The mobile agent with the shortest sensing duration will access the channel in the current slot by signaling immediately after the sensing duration. Omnidirectional antennas or sub-$6$GHz bands are used in signaling such that all other mobile agents can detect the channel occupation and suspend the channel access in the current slot. Moreover, it is assumed that the channel contention overhead is negligible compared to the remaining data transmission time of the slot.

The mean back-off time of the $k$-th mobile agent in the $t$-th slot is denoted as $1/\theta^{\texttt{A}}_{t,k}$, where $\theta^{\texttt{A}}_{t,k}\in
\mathfrak{I}\triangleq [\theta_{\texttt{min}},\theta_{\texttt{max}}]$  is referred to as the  back-off timer parameter. Consequently, it is proved in \cite{Modiano_Infocom2023} that the probability of the $k$-th mobile agent obtaining the access to the $t$-th slot can be expressed as
$
\upeta_{k}\left(\bm{\theta}^{\texttt{A}}_{t} \right)=\frac{\theta^{\texttt{A}}_{t,k}}{\sum_{\kappa \in \setK} \theta_{t,\kappa}^{\texttt{A}}}
$,
where $\bm{\theta}^{\texttt{A}}_{t} \triangleq (\theta_{t,k}^{\texttt{A}})_{k \in \setK}$. In this paper, the back-off timer parameters for all the slots and mobile agents, $(\theta_{t,k}^{\texttt{A}})_{t\in\mathbb{N},k \in \setK}$, will be optimized according to the urgency of data payload as well as the current channel status and future channel statistics. Thus, different packet arrivals or mobility patterns of mobile agents would lead to different back-off timer parameters.

\subsection{Uplink Channel and Transmission Model}
\label{section:channel_model}

A geometry-based channel model is adopted to characterize the location-dependent uplink channel considered in this paper. As illustrated in Fig. \ref{fig:system_model}, signal blockers and scattering clusters within the system coverage remain static. It has been demonstrated by the experiments \cite{Xinyu_Zhang_NSDI_2017,li2022indoor} that in such an environment, the number of paths, average power gains, and propagation path directions from any location to the AP exhibit temporal consistency. Hence, we define these parameters as functions of location in the following manner. Let $N_{\texttt {p}}(\loc )$ represent the number of paths from a given location $\loc\in \mathcal{L}$ to the AP,  $\theta_i(\loc)$ and $\phi_i({\bm \ell})$ denote the Angle-of-Arrival (AoA) and Angle-of-Departure (AoD) of the $i$-th path, $\upsigma^2_i({\bm \ell})$ signify its average power gain. The channel matrix $\mathbf{H}_{t,k}\in\mathbb{C}^{N{\texttt{R}}\times N_{\texttt{T}}}$ from the $k$-th mobile agent at the location $\loc_{t,k}^{\texttt{A}}$ to the AP in the $t$-th slot can be expressed as
\begin{small}
\begin{align*}
\mathbf{H}_{t,k}= \sqrt{N_{\texttt{T}}N_{\texttt{R}}} \!\!
\sum_{i=1}^{N_{\texttt {p}}(\loc_{t,k}^{\texttt{A}} )} \!\!\!\!
\alpha_{t,k,i}
\mathbf{a}_{\texttt{R}}\left( \phi_i \left(\loc^{\texttt A}_{t,k}  \right )  \right)
\mathbf{a}^{\mathsf{H}}_{\texttt{T}} \left( \theta_i \left(\loc^{\texttt A}_{t,k}  \right )  \right),
\end{align*}
\end{small}%
where $\alpha_{t,k,i}$ denotes the instantaneous path gain obeying a complex Gaussian distribution
$\alpha_{t,k,i} \sim \mathcal{CN} \left(0,\upsigma_{i}^2  \left( \loc^{\texttt A}_{t,k} \right)   \right)$. The instantaneous path gains remain quasi-static in one slot, and change independently in the next slot. Furthermore, the normalized array response vectors of ULAs at the mobile agent and the AP can be expressed as
\begin{small}
\begin{align*}
\mathbf{a}_{\texttt{T}}(\theta) &=\frac{1}{\sqrt{N_{\texttt{T}}}}\left[1,e^{-j\pi\sin(\theta)},\ldots,e^{-j\pi(N_{\texttt{T}}-1)\sin(\theta)}\right]^{\top},\\
\mathbf{a}_{\texttt{R}}(\phi) &=\frac{1}{\sqrt{N_{\texttt{R}}}}\left[1,e^{-j\pi\sin(\phi)},\ldots,e^{-j\pi(N_{\texttt{R}}-1)\sin(\phi)}\right]^{\top},
\end{align*}
\end{small}%
where $\theta$ and $\phi$ represent AoA and AoD respectively.

Let $\mathbf{w}_{t,k}\in\mathbb{C}^{N{\texttt{R}}\times 1}$ and $\mathbf{f}_{t,k}\in\mathbb{C}^{N{\texttt{T}}\times 1}$ denote the analog combiner and precoder employed for the uplink communication of the $k$-th mobile agent (if it wins the channel contention) in the $t$-th slot, the baseband power gain of this mobile agent can be expressed as
 $
 \SNR_{t,k} = {
 	\frac{ 1}{\sigma_{\texttt N}^2}	\left|\mathbf{w}_{t,k}^\mathsf{H}\mathbf{H}_{t,k}\mathbf{f}_{t,k}\right|^2}
$,
 where  $\sigma_\texttt{N}^2$ denotes the average noise power. In practice, the main lobe directions of $\mathbf{w}_{t,k}$ and $\mathbf{f}_{t,k}$ ($\forall t, k$) are not continuously adjustable. Instead, it is assumed that they are chosen from the following predefined candidate sets
 \begin{small}
 \begin{align}
 \mathbf{w}_{t,k}&\in\Big\{\mathbf{a}_{\texttt{R}}(\phi_{q}):\ q=1,2,\ldots,N_{\texttt{R}}\Big\}\triangleq\mathcal{W}, \\
 \mathbf{f}_{t,k}&\in\Big\{\mathbf{a}_{\texttt{T}}(\theta_{p}):\ p=1,2,\ldots,N_{\texttt{T}}\Big\}\triangleq\mathcal{F},
 \end{align}
 \end{small}%
 where
 $
 \phi_{q}=\arcsin\left( \frac{2(q-1)}{N_{\texttt{R}}}-1 \right)
 ,\label{eqn:combiner-angle}
 \theta_{p}=\arcsin\left( \frac{2(p-1)}{N_{\texttt{T}}}-1 \right)
 \label{eqn:procoder-angle}
 $. 
Furthermore, let $P_{\texttt{UL}}$ be the uplink transmit power. The  number of information bits transmitted from the $k$-th mobile agent to the AP in the $t$-th time slot (if it wins the channel contention) can be expressed as
$
R_{t,k}= T_{\texttt{slot}}W \log_{2}(1+ P_{\texttt{UL}} \SNR_{t,k})$,
where $W$ denotes the bandwidth, and $T_{\texttt{slot}}$ is the slot duration.

Conventional designs of the analog precoder and combiner  rely on the estimation of complete channel matrices $\left(\mathbf{H}_{t,k}\right)_{t,k}$, which can incur significant overhead, particularly with analog MIMO architecture. To mitigate this overhead, we shall rely only on the prior knowledge of the channel's geometric statistics, including  
$
 \left   \{N_{\texttt{p}}({\bm \ell}), \upsigma^2_i({\bm \ell}), \theta_i({\bm \ell}), \phi_i({\bm \ell}) : \ {\bm \ell} \in \mathcal{L},\ i = 1, 2, ..., N_{\texttt{p}}({\bm \ell}) \right \}
$.
It is shown in the experiments conducted in \cite{li2022indoor, Xinyu_Zhang_NSDI_2017} that these geometric statistics can be sensed via mmWave communication systems. Particularly, the following beamforming scheme based on the knowledge of geometric channel statistics and mobile agents' locations is adopted in this paper:
 	\vspace{-0.1cm}
 \begin{Scheme}[Location-Aided Beam Alignment]\label{Sch:Opt_LBA}
 Exploiting the location knowledge of mobile agents $\left(\loc_{t,k}^{\texttt{A}} \right) _{k \in \setK}$, the analog precoders and combiners for each agent-AP link are determined by the following optimization problem ($\forall t \in \mathbb{N},\ k \in \setK$)
 \begin{small}
 	\begin{align}	
 ( \mathbf{f}_{t,k},\    \mathbf{w}_{t,k}       )&=  \!\!\!\!\!\!
 	\mathop{\arg	\max}_{ \mathbf{f}_{t,k}\in\mathcal{F},\    \mathbf{w}_{t,k} \in \mathcal{W}   } \!\!\!\!\!\!\!
 	\mathds{E}_{\mathbf{H}_{t,k}}  \!\! \left[ 
 	\log_{2} \! \left( \! 1 \! +\!  \! \frac{  P_{\texttt{UL}}}{\sigma_{\texttt N}^2}	\left|\mathbf{w}_{t,k}^\mathsf{H}\mathbf{H}_{t,k}\mathbf{f}_{t,k}\right|^2 \! \right) \!
 	\right].\nonumber
 	\end{align}
  \end{small}%
 \end{Scheme}

\vspace{-0.1cm}
\subsection{System Queue Dynamics}
Each mobile agent is equipped with a dedicated uplink transmission queue, as depicted in Figure \ref{fig:system_model}. The arrival data at the mobile agents is organized into packets, each packet with $R_{\texttt{pac}}$ information bits. It is assumed that the number of arrival packets at the $k$-th mobile agent in the $t$-th slot, denoted as $A_{t,k}$, follows an independent Poisson distribution with an expectation value of $\bar{A}_{k}$, i.e.,
$
\mathds{P}[A_{t,k}=N]=
\frac{\left(\bar{A}_{k} \right)^N}{ N!}
e^{-\bar{A}_{k}}$.
 To streamline the analysis, we assume that all packets consistently arrive at the end of each time slot as in \cite{Neely2006}. 

Because of the randomness in the channel contention, the number of departure packets from the queue of the $k$-th mobile agent in the $t$-th slot is then given by
\begin{small}
\begin{align}
D_{t,k}=
\begin{cases}
\left\lfloor\frac{   R_{t,k} }{ R_{\texttt{pac}}}\right\rfloor, \quad &\textsf{with probability} \quad \upeta_{k}(\bm{\theta}^{\texttt{A}}_{t}) \\
0,\quad  &\textsf{with probability} \quad 1-\upeta_{k}(\bm{\theta}^{\texttt{A}}_{t})
\end{cases}
\end{align}
\end{small}%
Hence, let ${Q}_{\texttt{max}}$ denote the maximum buffer size, measured in terms of packets, for each uplink queue, $Q_{t,k}\in \mathcal{Q} \triangleq \{0,1,\ldots,{Q}_{\texttt{max}} \}$
represent the queue length of the $k$-th mobile agent at the beginning of the $t$-th slot. The dynamics of the uplink queues can be expressed as 
$
Q_{t+1,k}=
\min \left \{
\left( Q_{t,k}-D_{t,k} \right)^{+}+A_{t,k}, \ Q_{\texttt{max}} \right \},  \ \forall k \in \setK,
$
where arrival packets will be dropped if the uplink queue buffer is full.
\vspace{-0.2cm}
\section{Proposed Distributed Channel Contention Framework and Problem Formulation} \label{Sec:Problem_Formulation}

\begin{figure}[t]
	\centering
	\includegraphics[height=100pt,width=225pt]{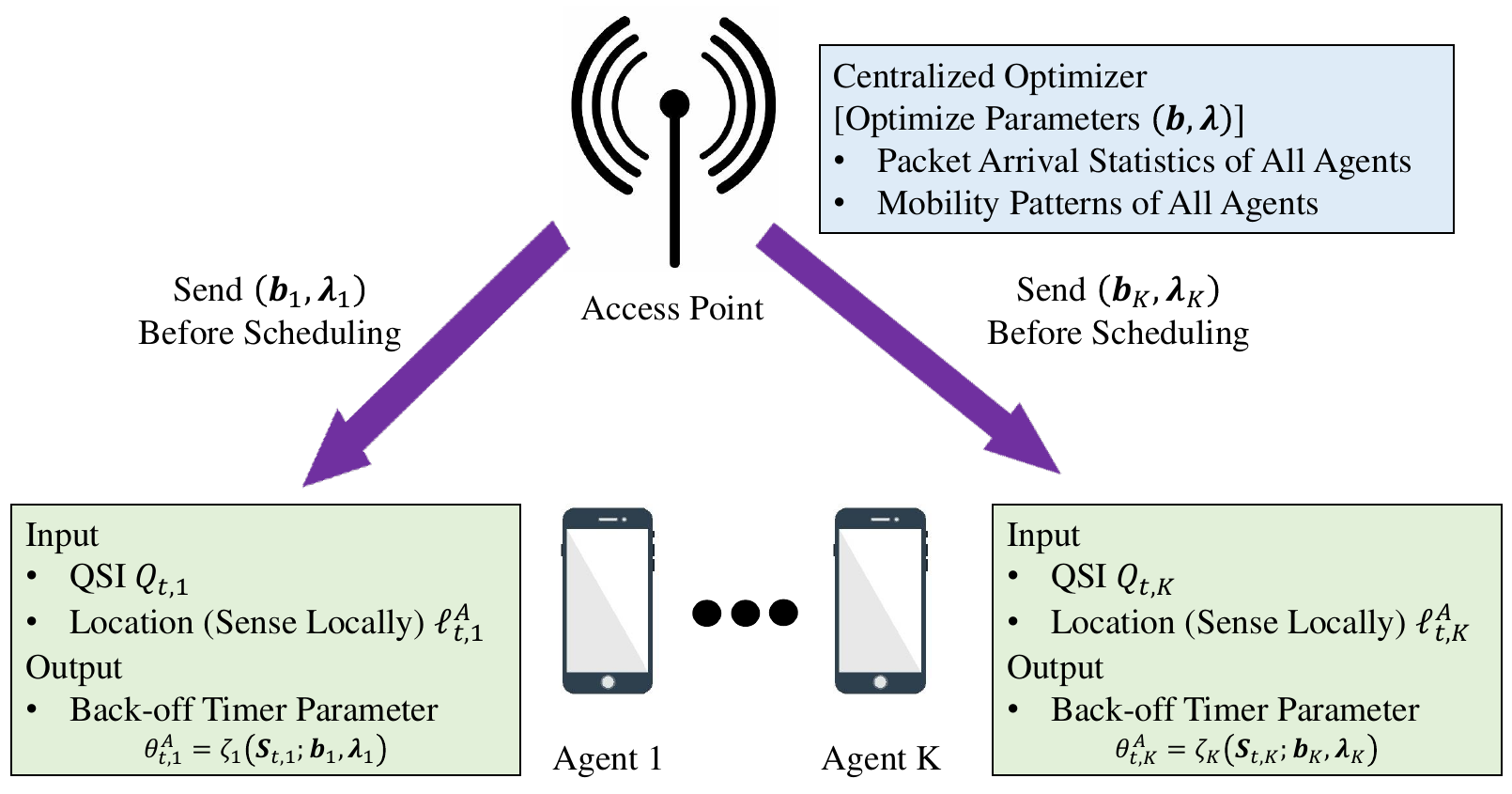}
	\caption{Illustration of proposed framework: centralized optimization
of distributed channel contention.}
	\label{fig:CTDE}
 	\vspace{-0.7cm}
\end{figure}

In the proposed distributed channel contention framework, each mobile agent (say the $k$-th mobile agent in the $t$-th slot) determines its back-off timer parameter $\theta^{\texttt{A}}_{t,k}$ based on its local system state, including local queue length $Q_{t,k}$ and location ${\bm \ell}_{t,k}^{\texttt A}$. Moreover, the optimization of the mapping from the local system states to their corresponding back-off timer parameters is conducted in a centralized manner at the AP beforehand based on the system statistics, including the channel statistics, queue dynamics and mobility patterns. 

Notice that given the above local mapping, the status of the considered mmWave communication system evolves as a Markov chain. In this section, we formulate the mapping optimization as a multi-agent decentralized MDP. First, the system states and local scheduling policies are defined as follows.
\vspace{-0.1cm}
\begin{Definition}[Local and Global System State]
		In the $t$-th slot ($\forall t \in \mathbb{N}$), the local system state of the mobile agent $k$ is characterized by the aggregation of its location and queuing state information (QSI), denoted as $
		\state_{t,k}\triangleq 
		\left( {\bm \ell}_{t,k}^{\texttt A}, \ Q_{t,k} \right)\in \mathcal{S} \triangleq \mathcal{L}\times
		\mathcal{Q}$
		;
the global system state is represented by the aggregation of all local system states, denoted as $\state_{t}\triangleq 
	\left( 	\state_{t,1}, 	\state_{t,2}, \ldots   , \state_{t,K} \right)\in  \mathcal{S}^{K} 
	$.
		
\end{Definition}
	\vspace{-0.3cm}

\begin{Definition}[Local Scheduling Policy]
 In the $t$-th slot ($\forall t \in \mathbb{N}$), the local scheduling policy of the mobile agent $k$ is a mapping from the local system state $\bm{\mathsf{S}}_{t,k}$ to the back-off timer parameter $\theta^{\texttt{A}}_{t,k}$. Thus,
		$
		\zeta_{k} \left( \state_{t,k} \right)=
		 \theta^{\texttt{A}}_{t,k}\in  \mathfrak{I}
		$.
\end{Definition}

Due to the huge global state and action spaces of the system, the optimal local policies for all the mobile agents is prohibitive. In this paper, we shall restrict the local policies to analytical models, such that the optimization can be accelerated and performance analysis can be facilitated. Particularly, let $\ell$ be the index of location point with $	 {\bm \ell}_{t,k}^{\texttt A}={\loc}^{ (\ell)}$, $b_{k,\ell}$ and $\lambda_{k,\ell}$ be the scheduling parameters for the $k$-th mobile agent at the location ${\bm \ell}^{(\ell)}$, $\mathbf{b}_{k}\triangleq (b_{k,i})_{i=1}^{ \left| \mathcal{L} \right | } \in \mathbb{R}^{  |\mathcal{L}| }_{+} $, and $\bm{\lambda}_{k}\triangleq ({\lambda}_{k,i})_{i=1}^{ \left| \mathcal{L} \right | }  \in \mathbb{R}^{  |\mathcal{L}| }_{+} $. The local scheduling policies are restricted to  
\begin{small}
\begin{align}
		 \theta^{\texttt{A}}_{t,k} = \zeta_{k} \left( \state_{t,k};\mathbf{b}_{k},\bm{\lambda}_{k} \right) \triangleq \proj_{\mathfrak{I}} ( b_{k,\ell}+ \lambda_{k,\ell}Q_{t,k} )\in  \mathfrak{I}, \label{eqn:policy_model}
\end{align}
\end{small}%
where the operator $\proj_{\mathfrak{I}}$ is defined as 
\begin{small}
    		\begin{align}
    		 \proj_{\mathfrak{I}}(x)\triangleq \max \left (\min \left(x, \theta_{\texttt{max}}  \right),  \theta_{\texttt{min}}   \right ), \ \forall x\in \mathbb{R}_{+}.
    		\end{align}
    		\end{small}%

Let $\mathbf{b}\triangleq \left ( \mathbf{b}_1 ^{\top} ,  \mathbf{b}_2 ^{\top},\ldots , \mathbf{b}_K ^{\top}  \right)^{\top}$ and $\bm{\lambda}\triangleq \left ( \bm{\lambda}_1 ^{\top} ,  \bm{\lambda}_2 ^{\top},\ldots , \bm{\lambda}_K ^{\top}  \right)^{\top}
$  be the aggregation vectors of local scheduling parameters.  The transition probabilities of the global system states can be represented as
\begin{small}
\begin{align}
\mathds{P}[ \state_{t+1}| \state_{t},\mathbf{b},\bm{\lambda} ]=&\prod_{k\in \setK} \mathds{P}[ \loc_{t+1,k}^{\texttt{A}}| \loc_{t,k}^{\texttt{A}} ] \mathds{P}[ Q_{t+1,k}| \state_{t},\mathbf{b},\bm{\lambda} ]. \label{eqn:GobalStateTransitionProb}
\end{align}
\end{small}%

According to Little's law, the average delay of uplink transmission is proportional to the average number of packets in the system \cite{1961A}. Hence, we first define the instantaneous system cost of the $t$-th slot given the system state $\state_{t}$ as the summation of the number of uplink packets buffered at mobile agents and the full buffer penalty as follows:
\begin{small}
\begin{align}\label{eqn:frame_level_cost}
\mathsf{c}_{\texttt{GS}} \left( 
\state_{t}
\right)
\triangleq
\sum_{k \in \setK} \mathsf{c}_{\texttt{LS}} \left( 
\state_{t,k}
\right),
\end{align}
\end{small}%
where $
\mathsf{c}_{\texttt{LS}} \left( 
\state_{t,k}
\right)\triangleq Q_{t,k}+w_{\texttt B}\mathds{1}[	Q_{t,k}=Q_{\texttt{max}}]
$
denotes the local cost of $k$-th mobile agent in the $t$-th slot, and $w_{\texttt B}$ denotes the weight of full buffer penalty. Then, the overall system cost is defined as the discounted summation of average system cost in all the slots. Thus,
\begin{small}
\begin{align}\label{eqn:obj}
\bar{\mathsf{G}}\left ( \state_{0};\mathbf{b},\bm{\lambda} \right)
\triangleq
\mathds{E}_{ (\bm{\mathsf{S}}_{t} )_{t\in \mathbb{N}} }
\left[
\sum_{t=0}^{+\infty}
\gamma^{t} \mathsf{c}_{\texttt{GS}}   \left( 
\state_{t}
\right)
\bigg|
\state_{0}, \mathbf{b},\bm{\lambda}
\right].
\end{align}
\end{small}%
As a result, the scheduling of the uplink transmission in this paper can be formulated as the following multi-agent MDP:
\begin{small}
\begin{align}
\textsf{P1}:\ 
(\mathbf{b}^{\star}, \bm{\lambda}^{\star})
\triangleq&\mathop{\arg\min}_{
	\mathbf{b}, \bm{\lambda}
} 
\bar{\mathsf{G}}\left(\state_{0}; \mathbf{b}, \bm{\lambda}  \right ).
\end{align}
\end{small}%
As illustrated in Fig. \ref{fig:CTDE},  the optimization of parameters $(\mathbf{b}, \bm{\lambda})$ is centralized. After optimization, each agent receives its personalized set
of optimized parameters $(\mathbf{b}_{k}, \bm{\lambda}_{k})$ and establishes a local
policy $ \zeta_{k} \left( \state_{t,k};\mathbf{b}_{k},\bm{\lambda}_{k} \right)$.

Given the local scheduling parameters $ \left(\mathbf{b}, \bm{\lambda} \right) $, the global system states transit according to the probabilities defined in \eqref{eqn:GobalStateTransitionProb}. Different values of $ \left(\mathbf{b}, \bm{\lambda} \right) $ would lead to different transition probabilities, and hence different overall system costs. The optimization on them is feasible as long as the knowledge on the \eqref{eqn:GobalStateTransitionProb} is known. However, problem  $\textsf{P1}$ is still challenging. On the one hand, the gradient of the objective depends on the transition probabilities of all system states, which is prohibitive in calculation due to the massive space of global system states. On the other hand, the efficient stochastic gradient descent method is not straightforward. For example, it is unclear how many transition trails of global system states are necessary to obtain an unbiased estimation of overall system cost's gradient. In the following section, an unbiased gradient estimation is derived based on a single trail of global system states. As a result, problem  $\textsf{P1}$ can be efficiently solved via the stochastic gradient descent method. 

\vspace{-0.1cm}
 \begin{Remark}
 The analytical models for back-off timer parameters in (\ref{eqn:policy_model}) are monotonically increasing with respect to the queue length. This is because larger queue length leads to larger chance of full buffer, and hence higher transmission priority. Moreover, mobile agents at different locations should have different transmission priorities, which is realized by adjusting the local scheduling parameters $ \left(\mathbf{b}, \bm{\lambda} \right) $. As a remark note that although a truncated linear model is used for back-off timer parameters in this paper, our method can also be applied on other models.
 	\vspace{-0.4cm}
\end{Remark}

\begin{figure*}
\begin{small}
	\begin{align}
	\hat{\mathsf{g}}^{\texttt{b}}_{\kappa,\ell}  \left( \mathbf{b}, \bm{\lambda}; \traj  \right ) 
	&\triangleq \left[ 
	\left (   \sum_{t \in \mathbb{N}} 
	\gamma^{t} \mathsf{c}_{\texttt{GS}}   
	\left(\state_{t}^{\mathbf{b}, \bm{\lambda}} \right)
	\right)
 \left(
	\sum_{t \in \mathbb{N}}
	\sum_{ k\in\setK} 
	\frac{ \partial \log\mathds{P}_{k}\left[\state_{t+1,k}^{\mathbf{b}, \bm{\lambda}} \Big| \state_{t}^{\mathbf{b}, \bm{\lambda}}  ,\mathbf{b}, \bm{\lambda} \right] }{ \partial  b_{\kappa,\ell} } \right) \right].  \label{eqn:grad_b}\\
\hat{	\mathsf{g}}^{\lambda}_{\kappa,\ell}  
\left( \mathbf{b}, \bm{\lambda}; \traj \right ) 
	&\triangleq  \left[ 
	\left (   \sum_{t \in \mathbb{N}}
	\gamma^{t} \mathsf{c}_{\texttt{GS}}  
		\left(\state_{t}^{\mathbf{b}, \bm{\lambda}} \right)
	\right)
\left(
	\sum_{t \in \mathbb{N}}
	\sum_{ k\in\setK}
	\frac{ \partial \log\mathds{P}_{k}\left[\state_{t+1,k}^{\mathbf{b}, \bm{\lambda}} \Big| \state_{t}^{\mathbf{b}, \bm{\lambda}}  ,\mathbf{b}, \bm{\lambda} \right] }{ \partial  {\lambda}_{\kappa,\ell} } \right) \right]. \label{eqn:grad_lambda}\\
\frac{ \partial \log\mathds{P}_{k}\left[\state_{t+1,k}^{\mathbf{b}, \bm{\lambda}} \Big| \state_{t}^{\mathbf{b}, \bm{\lambda}}  ,\mathbf{b}, \bm{\lambda} \right] }{ \partial  {b}_{\kappa,\ell} } 
&=\frac{{\upomega}_{t,k}^{(1)}-{\upomega}_{t,k}^{(2)}}{\upeta_{k}(\bm{\theta}(\mathbf{b}, \bm{\lambda}))\left({\upomega}_{t,k}^{(1)}-{\upomega}_{t,k}^{(2)}\right)+{\upomega}_{t,k}^{(2)}}
	\frac{ \partial  \upeta_{k}(\bm{\theta}(\mathbf{b}, \bm{\lambda}))}{ \partial  {b}_{\kappa,\ell}  }. \label{eqn:dlog/db} \\
	\frac{ \partial \log\mathds{P}_{k}\left[\state_{t+1,k}^{\mathbf{b}, \bm{\lambda}} \Big| \state_{t}^{\mathbf{b}, \bm{\lambda}}  ,\mathbf{b}, \bm{\lambda} \right] }{ \partial  {\lambda}_{\kappa,\ell} }  &=\frac{{\upomega}_{t,k}^{(1)}-{\upomega}_{t,k}^{(2)}}{\upeta_{k}(\bm{\theta}(\mathbf{b}, \bm{\lambda}))\left({\upomega}_{t,k}^{(1)}-{\upomega}_{t,k}^{(2)}\right)+{\upomega}_{t,k}^{(2)}}
	\frac{ \partial  \upeta_{k}(\bm{\theta}(\mathbf{b}, \bm{\lambda})  )}{ \partial  {\lambda}_{\kappa,\ell}  }. \label{eqn:dlog/dlambda}\\
	b_{\kappa,\ell}(m+1)&=b_{\kappa,\ell}(m)- \eta_{\texttt{LPG}}(m)
\hat{	\mathsf{g}}^{\texttt{b}}_{\kappa,\ell}  
	\left( \mathbf{b}(m), \bm{\lambda}(m); 
	\bm{\tau}^{\mathbf{b}(m), \bm{\lambda}(m) }_{\texttt{MC}  } 
	\right)
\label{eqn:update_b}. \\ 
	\lambda_{\kappa,\ell}(m+1)&={\lambda}_{\kappa,\ell}(m)- \eta_{\texttt{LPG}}(m)
	\hat{\mathsf{g}}^{\lambda}_{\kappa,\ell}  
	\left( \mathbf{b}(m), \bm{\lambda}(m); 
	\bm{\tau}^{\mathbf{b}(m), \bm{\lambda}(m) }_{\texttt{MC}  } 
	\right). \label{eqn:update_lambda}
	\end{align}
 \end{small}%
		\hrulefill
\end{figure*}

\section{Proposed SGD-BASED SOLUTION}
\label{Sec:Stochastic_Localized_Policy_Gradient_Descent}
Because the objective defined in \eqref{eqn:obj} traverses through all possible global system states $(\bm{\mathsf{S}}_{t})_{t\in \mathbb{N}}$,
the computation of gradients $	\frac{ \partial\bar{\mathsf{G}}(\state_{0}; \mathbf{b}, \bm{\lambda} )  }{ \partial b_{k,\ell}}$ and $	\frac{ \partial \bar{\mathsf{G}}(\state_{0}; \mathbf{b}, \bm{\lambda} )  }{ \partial {\lambda}_{k,\ell}}$ becomes prohibitive. In order to facilitate the optimization of $(\mathbf{b}, \bm{\lambda})$, unbiased estimation of these gradients is first proposed in this part, followed by a stochastic gradient descent method. 

Particularly, let  $
\traj \triangleq
\left(\state_{0}, \state_{1}^{\mathbf{b}, \bm{\lambda}}, 
\state_{2}^{\mathbf{b}, \bm{\lambda}},
\ldots \right)$ be a sequence of observed system states with the local scheduling parameters $(\mathbf{b}, \bm{\lambda})$, 
We have the following conclusion on the unbiased observation of the gradients. 
\begin{Lemma}\label{lem:Exp-PG}
Given the local scheduling parameters $(\mathbf{b}, \bm{\lambda})$,
$	\hat{\mathsf{g}}^{\texttt{b}}_{\kappa,\ell}  \left( \mathbf{b}, \bm{\lambda}; \traj \right ) $ and $\hat{	\mathsf{g}}^{\lambda}_{\kappa,\ell}  
\left( \mathbf{b}, \bm{\lambda}; \traj \right ) $ defined in \eqref{eqn:grad_b} and \eqref{eqn:grad_lambda} are the unbiased estimation of  $	\frac{ \partial\bar{\mathsf{G}}(\state_{0}; \mathbf{b}, \bm{\lambda} )  }{ \partial b_{\kappa,\ell}}$ and $	\frac{ \partial \bar{\mathsf{G}}(\state_{0}; \mathbf{b}, \bm{\lambda} )  }{ \partial {\lambda}_{\kappa,\ell}}$, $\forall \kappa,\ell$, respectively.

\end{Lemma}

\begin{proof}
	Please refer to \cite{SM}.
		\vspace{-0.4cm}
\end{proof}

Moreover, the expression of  $\frac{ \partial \log\mathds{P}_{k}\left[\state_{t+1,k}^{\mathbf{b}, \bm{\lambda}} | \state_{t}^{\mathbf{b}, \bm{\lambda}} ,\mathbf{b}, \bm{\lambda} \right] }{ \partial {b}_{\kappa,\ell} }$ in \eqref{eqn:grad_b} and $\frac{ \partial \log\mathds{P}_{k}\left[\state_{t+1,k}^{\mathbf{b}, \bm{\lambda}} | \state_{t}^{\mathbf{b}, \bm{\lambda}} ,\mathbf{b}, \bm{\lambda}\right] }{ \partial {\lambda}_{\kappa,\ell} }$  in \eqref{eqn:grad_lambda} are provided in \eqref{eqn:dlog/db} and \eqref{eqn:dlog/dlambda}, respectively, and the expressions for ${\upomega}_{t,k}^{(1)}$, ${\upomega}_{t,k}^{(2)}$,  $ \frac{ \partial \upeta_{k}(\bm{\theta})}{ \partial b_{\kappa,\ell} }$, and $ \frac{ \partial \upeta_{k}(\bm{\theta})}{ \partial {\lambda}_{\kappa,\ell} }$ are provided in \cite{SM}.

As a result, the SGD method solving the problem $\textsf{P1}$ is outlined in Algorithm \ref{Alg:SLPG}. Finally, the following theorem establishes the convergence of Algorithm \ref{Alg:SLPG}.
\begin{Theorem}[Stochastic Convergence] \label{Thm:Convergence_SGD}
Let $\om \in \mathbb{R}_{+}$ be a constant such that
\begin{small}
\begin{align*}
\sum_{k=1}^{K}\sum_{\ell=1}^{|\mathcal{L}|} &\mathds{E}\left[ \left(\hat{\mathsf{g}}^{\texttt{b}}_{k,\ell}  
\left( \mathbf{b}, \bm{\lambda}; 
\traj
\right)\right)^2+
\left(
\hat{\mathsf{g}}^{\lambda}_{k,\ell} 
\left( \mathbf{b}, \bm{\lambda}; 
\traj
\right) \right)^2 \right]\leq \om,
\end{align*}
\end{small}%
we have
\begin{small}
\begin{align*}
&\mathds{E} \left[\min_{m\in \left\llbracket 0:M\right \rrbracket} \eta_{\texttt{LPG}}(m) \left\Vert \nabla \bar{\mathsf{G}} \left (\state_{0}; \mathbf{b}(m),\bm{\lambda}(m) \right)  \right\Vert ^{2}_{2} \right]\nonumber\\
&\leq \frac{1}{M+1}  \mathds{E} \left[ \bar{\mathsf{G}}\left(\state_{0}; \mathbf{b}(0),
\bm{\lambda}(0) \right)
-\bar{\mathsf{G}} \left(\state_{0}; \mathbf{b}^{\star}
, \bm{\lambda}^{\star}
\right) + \frac{\pi^2}{12} \om
\right].
\end{align*}
\end{small}%
\end{Theorem}
\begin{proof}
	Please refer to \cite{SM}.
		\vspace{-0.3cm}
\end{proof}
As $M$ tends to infinity, the convergence of the proposed SGD-based Algorithm is assured by Theorem 1.



\begin{algorithm}

	\SetKwInOut{Input}{input}\SetKwInOut{Output}{output}
	
	\Input{ $(\mathbf{b}(0), \bm{\lambda}(0) )$; $T_{\texttt{ep}}$, length of each episode; step size parameter $\eta_{\texttt{LPG}}(0)$.   }
	\Output{$(\mathbf{b}(\infty), \bm{\lambda}(\infty) )$.}
	\BlankLine
	\For{$m=0,1,2,\dots$ }{
		
		Initial global state $\state_0 $, receive global system cost $\mathsf{c}_{\texttt{GS}}(\state_0)$, each agent $k$ takes action $\theta_{0,k}^{\texttt A}=\zeta_{k}\left( \state_{0,k} ;\mathbf{b}(m), \bm{\lambda}(m)  \right)$.

		\For{$t=1$ \KwTo $T_{\texttt{ep}}+1$}
		{\label{forins}
			
			Get global state $\state_{t}^{\mathbf{b}(m), \bm{\lambda}(m) }$ and global system cost $\mathsf{c}_{\texttt{GS}}\left(\state_{t}^{^{\mathbf{b}(m), \bm{\lambda}(m) }}\right)$, each agent $k$ takes action $\theta_{t,k}^{\texttt A}=\zeta_{k}\left( \state_{t,k}^{\mathbf{b}(m), \bm{\lambda}(m) } ;\mathbf{b}(m), \bm{\lambda}(m)  \right)$.

		}
		For any $\kappa\, \ell$, caculate gradients
		$\hat{\mathsf{g}}^{\texttt{b}}_{\kappa,\ell} \left ( \mathbf{b}(m), \bm{\lambda}(m); 
		\bm{\tau}^{\mathbf{b}(m), \bm{\lambda}(m) }_{\texttt{MC}  } 
		\right)$ as \eqref{eqn:grad_b}, and $	\hat{\mathsf{g}}^{\lambda}_{\kappa,\ell}  \left( \mathbf{b}(m), \bm{\lambda}(m); 	\bm{\tau}^{\mathbf{b}(m), \bm{\lambda}(m) }_{\texttt{MC}  }  \right) $ as \eqref{eqn:grad_lambda}

	For any $\kappa, \ell $, 	update parameters as \eqref{eqn:update_b} and \eqref{eqn:update_lambda}, where $\eta_{\texttt{LPG}}(m)= \frac{\eta_{\texttt{LPG}}(0)}{m+1}$.
	}
	\caption{SGD-Based Solution}\label{Alg:SLPG}
		
\end{algorithm}

\begin{table}\footnotesize%
	\centering
	\begin{tabular}{|c|c|c|}
		\hline
		\textbf{Parameter} & \textbf{Symbol} & \textbf{Value}\\
\hline		
  Number of mobile agents& $ K $ & 8 \\
  \hline
		Number of receive/transmit antenna  & $N_{\texttt{R}}$, $N_{\texttt{T}}$ & $64$, $64$ \\
		\hline
		Slot duration & $T_{\texttt{F}}$ & $3.008 \  \text{ms}$ \\
		\hline
		Max. /Min. back-off timer parameter  & ${\theta}_{\texttt{max}}$, ${\theta}_{\texttt{min}} $ & $1$, $\frac{1}{63}$ \\
		\hline
		Number of information bits in one packet  & ${R}_{\texttt{pac}}$ & $1 \ \text{M}$ \\
		\hline
		Discount factor & $\gamma$ & $0.95$\\
		\hline
		Max. buffer size in terms of packets & $Q_{\texttt{max}}$ & $10$  \\
		\hline
		Uplink power  & $P_{\texttt{UL}}$ & 
		$1\ \text{Watts}$\\
		\hline
	\end{tabular}
	\caption{Parameters of the simulation.}	\label{tab:parameter}
 	\vspace{-0.6cm}
\end{table}

\section{Simulations and Discussions}\label{Sec:Simulations}

In this section, the effectiveness of the proposed SGD method is illustrated by simulation. The main simulation parameters are summarized in Table \ref{tab:parameter}.
Similar to \cite{maltsev2010channel,maltsev2016channel}, we simulate the delay-sensitive communication in a $9$m × $7.5$m square room, divided into grids as depicted in Fig. \ref{fig:indoor_simulation}. The AP is at $(4.5, 0)$, and three static blockers with a radius of $0.5$m are centered at $(1.5, 2)$, $(4.5, 2)$, and $(7.5, 2)$ respectively. The indoor space is enclosed by four walls, and the specular reflection off the walls is considered in the mmWave channel. In each time slot, a mobile agent may either remain static or move to one of the four neighboring grids with equal probability.

In the simulation, the uplink channels are generated as follows. In each time slot, given the positions of the blockers, the AP, and the mobile agents, the propagation paths from the mobile agents to the AP are first determined. Their AoAs, AoDs, and average power gains (or losses) can then be calculated. We follow the pathloss and reflection loss model outlined in \cite{maltsev2010channel,maltsev2016channel}. Subsequently, the instantaneous path gains are generated using a complex Gaussian distribution.

The proposed algorithm's performance is compared with the following five baseline policies.

\begin{figure}[t]
	\centering
	\includegraphics[height=120pt,width=225pt]{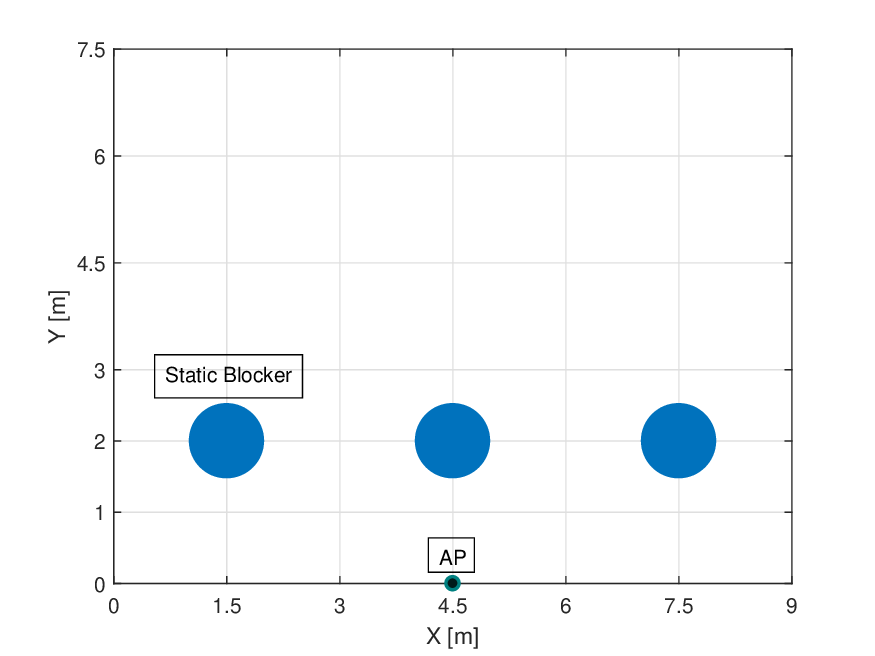}
	\caption{Indoor simulation scenario.}
	\label{fig:indoor_simulation}
 	\vspace{-0.6cm}
\end{figure}

\begin{figure}[t]
	\centering
	\includegraphics[height=120pt,width=225pt]{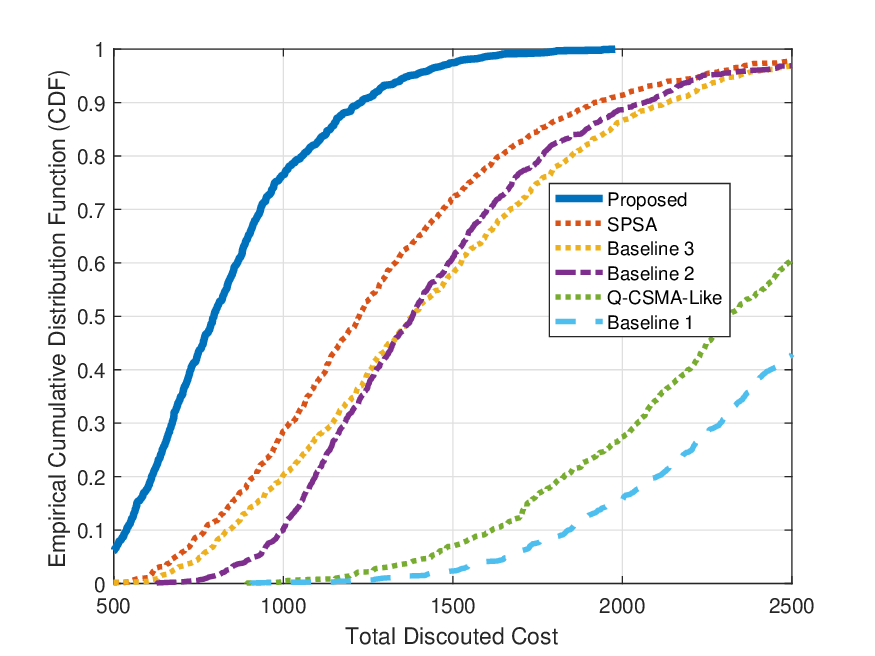}
	\caption{The empirical cumulative distribution function (CDF) of the total discounted cost in the indoor environment, where
	average arrival rate $\bar{A}_{k}=0.6$, $\forall k\in \setK$. }
	\label{fig:CDF_indoor}
	\vspace{-0.65cm}
\end{figure}

\begin{figure}[t]
	\centering
	\includegraphics[height=120pt,width=225pt]{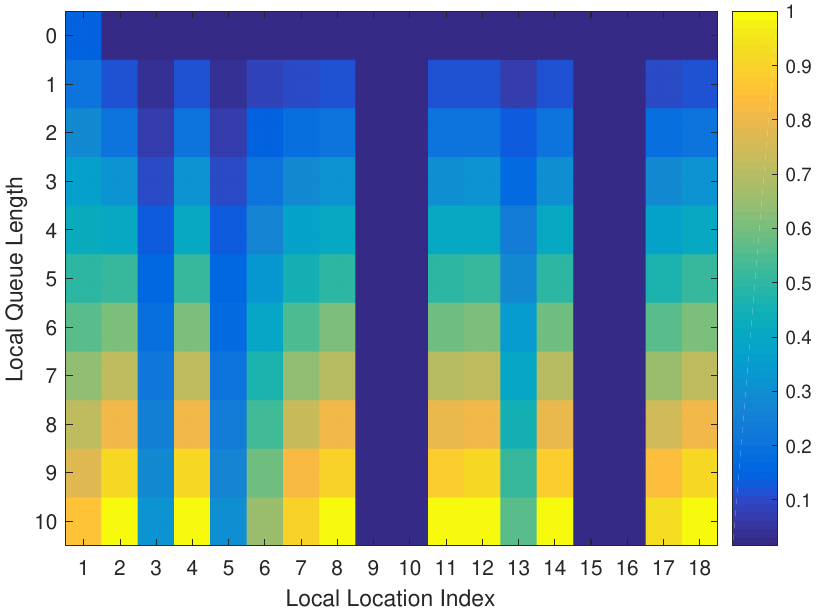}
	\caption{Visualization of an agent's proposed back-off timer parameters corresponding to all local states, where average arrival rate $\bar{A}_{k}=0.6$, $\forall k\in \setK$.}
	\label{fig:acton}
\vspace{-0.7cm}
\end{figure}

\begin{itemize}

\item {\bf Baseline 1:} 
Each agent chooses a constant back-off timer parameter, i.e.,
$
	    	\zeta_{k}(\state_{t,k})= \frac{1}{2} \left ( 
		\theta_{\texttt{min}}+ \theta_{\texttt{max}   }
		\right), \ \forall \state_{t,k} \in \mathcal{S}.
$

\item {\bf Baseline 2:} 
Each agent selects a back-off timer parameter based on its local queue length. If the uplink queue is full, the agent chooses $\theta_{\texttt{max}}$. $\theta_{\texttt{min}}$ is selected otherwise, i.e.,
\begin{align}
		    \zeta_{k}(\state_{t,k})=
		    \begin{cases}
	\theta_{\texttt{max}}, \quad  &Q_{t,k} = Q_{\texttt{max}} \\
	\theta_{\texttt{min}},\quad  &\textsf{others}
\end{cases}
\end{align}
		
\item {\bf Baseline 3:} Each agent chooses a back-off timer parameter based on the localized policy
\begin{align*}
     \zeta_{k} \left( \state_{t,k};\mathbf{b}_{k}(0),\bm{\lambda}_{k}(0) \right)=\proj_{\mathfrak{I}} \left[ b_{k,\ell}(0)+ \lambda_{k,\ell}(0) Q_{t,k}  \right],
\end{align*}
where ${b}_{k,\ell}(0)={\theta}_{\texttt{min}}$, ${\lambda}_{k,\ell}(0)=\frac{ {\theta}_{\texttt{max}}-
{\theta}_{\texttt{min}}}{Q_{\texttt{max}}}$ for any $k,\ell$.

\item {\bf Q-CSMA-Like:}
Similar to \cite{Srikant_2012ToN}, each agent selects a back-off timer parameter according to its local queue length, i.e.,
$
     \zeta_{k}(\state_{t,k})=\proj_{\mathcal{I}} \left[ \frac{\log (Q_{t,k})}{1+\log (Q_{t,k})} \right ].
$

\item {\bf{Simultaneous Perturbation Stochastic Approximation (SPSA)} \cite[Algorithm 17]{krishnamurthy2016partially} :}
Each agent selects a back-off timer parameter according to the localized policy
$
		\zeta_{k} \left( \state_{t,k};\mathbf{b}_{k}^{\dagger},\bm{\lambda}_{k}^{\dagger} \right)=\proj_{\mathfrak{I}} \left[ b_{k,\ell}^{\dagger}+ \lambda_{k,\ell} ^{\dagger}Q_{t,k}  \right],
$
where
$\left ( \mathbf{b}_{k}^{\dagger},\bm{\lambda}_{k}^{\dagger} \right )$ represents the outcome after optimization using SPSA.
\end{itemize}
		\vspace{-0.1cm}
	
Note that the initial parameters $\left( \mathbf{b}_{k}(0), \bm{\lambda}_{k}(0) \right)_{k \in \setK}$
for the proposed method and SPSA are identical to those used in Baseline 3.

First of all, the cumulative distribution functions (CDFs) of system cost are compared in  Fig. \ref{fig:CDF_indoor}. Significant improvement in our proposed method compared to the five baselines can be observed. The SPSA scheme also outperforms the other baselines. Both performance gains demonstrate the necessity of optimizing back-off timer parameters. The performance gain of the proposed method over the SPSA scheme illustrates the effectiveness of the proposed gradient estimation.

Figure \ref{fig:acton} visualizes an agent's proposed back-off timer parameters. The color intensity in the graph indicates the values of the selected back-off timer parameter $\theta_{t,k}^{\texttt{A}}$. A darker color represents a smaller $\theta_{t,k}^{\texttt{A}}$, and consequently a lower probability of uplink transmission. In areas with LoS obstruction (indexed as $3$, $4$, $9$, $10$, $15$, and $16$ in the x-axis), the agent chooses a small $\theta_{t,k}^{\texttt{A}}$ to relinquish uplink resources to other agents. This distributed approach ensures a rational resource allocation among multiple agents.

	\vspace{-0.2cm}
\section{Conclusions}\label{Sec:Conclusion}

This paper presents a decentralized multi-agent MDP framework for optimizing distributed and adaptive channel contention in mmWave-based systems. Specifically, mobile agents' local channel contention actions depend solely on their local statuses. The optimization of these local policies for all mobile agents is conducted in a centralized manner based on system statistics before scheduling. In our approach, the local policies are approximated using analytical models, and optimizing their parameters poses a stochastic optimization problem. Due to the prohibitive computational complexity of gradients, we introduce an SGD-based method and analyze its convergence rate for a specific step size. Our proposed solution integrates LoS blockage prediction into gradient estimation, allowing each slot's scheduling to account for future buffer overflow risks. Simulation results demonstrate that our scheme effectively reduces queuing delays and buffer overflow rates caused by LoS blockage, outperforming baseline approaches in performance

	\vspace{-0.4cm}
\bibliographystyle{IEEEtran}
\bibliography{mmWave_Decentralize}
\end{document}